\documentclass[12pt]{iopart}
\usepackage{amssymb,amsthm,graphicx}
\usepackage{graphicx}
\usepackage{colordvi}
\usepackage{color}
\usepackage{iopams}

\newcommand{\sgn}{\mathrm{sgn}\,}
\newcommand{\eps}{\varepsilon}

\newcommand{\OO}{\mathcal{O}}

\newtheorem{claim}{Claim}[section]
\newtheorem{thm}[claim]{Theorem}
\newtheorem{prop}[claim]{Proposition}

\newtheorem{cor}[claim]{Corollary}
\begin{document}

\title[Spectra of magnetic chain graphs]
{Spectra of magnetic chain graphs: coupling constant perturbations}

\author{Pavel Exner}
\address{Doppler Institute for Mathematical Physics and Applied
Mathematics, \\ Czech Technical University in Prague,
B\v{r}ehov\'{a} 7, 11519 Prague, \\ and  Nuclear Physics Institute
ASCR, 25068 \v{R}e\v{z} near Prague, Czechia} \ead{exner@ujf.cas.cz}

\author{Stepan S. Manko}
\address{Department of Physics, Faculty of Nuclear Science and Physical Engineering, \\ Czech Technical University in Prague,
Pohrani\v{c}n\'{i} 1288/1,  40501 D\v{e}\v{c}\'{i}n,
Czechia\footnote{On leave of absence from Pidstryhach Institute for Applied Problems of Mechanics and Mathematics, National Academy of Sciences of Ukraine, 3b Naukova str, 79060 Lviv, Ukraine}} \ead{stepan.manko@gmail.com}

\begin{abstract}
We analyze spectral properties of a quantum graph in the form of a ring chain with a $\delta$ coupling in the vertices exposed to a homogeneous magnetic field perpendicular to the graph plane. We find the band spectrum in the case when the chain exhibits a translational symmetry and study the discrete spectrum in the gaps resulting from changing a finite number of vertex coupling constants. In particular, we discuss in details some examples such as perturbations of one or two vertices, weak perturbation asymptotics, and a pair of distant perturbations.
\end{abstract}

% Keywords: hedgehog manifolds, Weyl asymptotics, quantum graphs, resonances.
%\noindent{\rmit Keywords}: Article preparation, IOP journals
\maketitle

%%%%%%%%%%%%%%%%%%%%%%%%%%%%%%%%%%%%%%%%%%%%%%%%%%%%%%%%%%
%%  INTRODUCTION                                        %%
%%%%%%%%%%%%%%%%%%%%%%%%%%%%%%%%%%%%%%%%%%%%%%%%%%%%%%%%%%

\section{Introduction}

Quantum graphs keep attracting a lot of attention, both as simple but effective models of numerous artificially prepared microscopic structures in which transport is dominantly ballistic, as well as a source of interesting questions in quantum theory. We refer to the recent monograph \cite{BK13} which can give the reader a picture of the breath and richness of the subject. In this paper we focus on one class of such graphs, namely those with a \emph{chain structure}. They are characterized by a mixed dimensionality being periodic in one direction, thus giving typically rise to band-and-gap spectra, the structure one which depends on the `transverse' shape, see e.g. \cite{EKW10}. Moreover, it has been observed recently that manipulating the vertex coupling one can control the transport in such chains in an intriguing way \cite{CP14}.

A question of a particular interest, both theoretically and from the practical point of view, concerns influence of magnetic field on quantum dynamics on such graph. Our aim in this paper is to study the simplest chain consisting of a single array of rings in presence of a homogeneous magnetic field; in contrast to \cite{CP14} we consider the simplest nontrivial vertex coupling. Our main concern will be local perturbations of periodicity and the discrete spectrum they cause. A similar problem in the nonmagnetic case was considered in \cite{DET08} where the influence of a simple geometric perturbations, `bending' of the chain, was analyzed. Here we focus on perturbations coming from coupling constant variations. It is not the only possibility, of course, in a sequel to this paper we are going to discuss the effect of local field modifications.

Let us now describe the model and the questions to be addressed in more details. We are going to consider a chain graph $\Gamma$ consisting of an array of rings of unit radius --- cf. Fig.~\ref{fig:graph} --- connected through their touching points. The graph $\Gamma$ is naturally parametrized by two copies of the real line $\mathbb{R}$ corresponding to the upper and lower semicircles, respectively. The state Hilbert space of a nonrelativistic and spinless charged particle confined to $\Gamma$ is $L^2(\Gamma)$. As indicated, we suppose that the particle is moving under the influence of a homogeneous magnetic field perpendicular to the graph plane. Since values of the physical constants are not important in our considerations we put $\hslash=2m=e=1$, where $e$ is the particle charge, and identify the particle Hamiltonian with the magnetic Laplacian acting as $\psi_j\mapsto-\mathcal{D}^2\psi_j$ on each graph link (the definition of the  quasiderivative $\mathcal{D}$ will be given in the next section). The domain of this operator consists of all functions from the Sobolev space $H^2_\mathrm{loc}(\Gamma)$ satisfying the $\delta$-coupling at the graph vertices which are characterized by conditions
 % -------------- %
\begin{equation}
\label{cond:deltaCoupling1}
\psi_i(0) = \psi_j(0) = :\psi(0)\,,
\quad
i,j\in\mathfrak{n}\,,
\qquad
\sum_{i=1}^n
\mathcal{D}\psi_i(0)
=\alpha\,\psi(0)\,,
\end{equation}
 % -------------- %
where $\mathfrak{n}=\{1,2,\ldots,n\}$ is the index set numbering the edges emanating from the vertex --- in our case $n=4$ --- and $\alpha\in\mathbb{R}\cup\{\infty\}$ is the coupling constant possibly different at different vertices of the chain. Note the vertex coupling (\ref{cond:deltaCoupling1}) is a particular --- and simplest nontrivial --- case of the general conditions that make the graph magnetic Laplacian self-adjoint \cite{KS03}.

In what follows we denote by $\balpha=\{\alpha_j\}_{j\in\mathbb{Z}}$ the set of coupling constants and by $-\Delta_{\balpha}$ the above described Hamiltonian. In accordance with our stated goal, we start the discussion of spectral properties of this operator in Sec.~\ref{sec:period} with the periodic case, i.e. the situation when all the $\alpha_j$ in $\balpha$ are equal to fixed $\alpha\in\mathbb{R}$. It is not difficult to perform the Bloch-Floquet analysis showing that in this case the spectrum of $\sigma(-\Delta_{\balpha})$ has a band-and-gap structure.

As we have indicated, our main goal is to analyze coupling constant perturbations, that is, to look what happens if the `constant' sequence $\{\ldots,\alpha,\alpha,\ldots\}$ is modified to
 % -------------- %
\begin{eqnarray*}
\alpha_j &= \alpha+\gamma_j\,,\qquad &j\in\mathbb{M}:=\{1,\ldots,m\}\,,
\\
\alpha_j &= \alpha\,,\qquad &j\in\mathbb{Z}\setminus\mathbb{M}\,.
\end{eqnarray*}
 % -------------- %
The method we employ is based on translating the spectral problem for the differential equation in question into suitable difference equations. This trick is well known in quantum graph theory \cite{Ca97, Ex97, Pa13}, however, it is assumed usually that not more than one edge connects any two graph vertices. Since this is not the case here, we have to work out the procedure for our purpose; this will be done in Sec.~\ref{sec:duality}. Using it we shall analyze in Sec.~\ref{sec:mcoupl} the discrete spectrum coming from the perturbation both generally and in examples where one or two vertices are perturbed.

The last two sections are devoted to particular cases. One is motivated by comparison to the usual theory of Schr\"{o}dinger operators; we ask about sufficient conditions on the perturbation $\gamma_j$ to assure existence of an eigenvalue of the perturbed operator in the first spectral gap (under the threshold of the continuous spectrum) of the unperturbed system. In Sec.~\ref{sec:wcoupl} we will show that such a condition has the form
 % -------------- %
\begin{equation}\label{cond:eigValExist}
\sum_{j\in\mathbb{M}}\gamma_j < 0\,,
\end{equation}
 % -------------- %
which is an analogue of the well-known necessary and sufficient condition $\int_\mathbb{R} V(x)\,\rmd x\le 0$ valid for 1D Schr\"{o}dinger operators $-\rmd^2/\rmd x^2+V(x)$ with potentials $V$ satisfying certain integral-form decay requirements \cite{Si76}. If the perturbation is sufficiently weak, there is precisely one eigenvalue below the continuous spectrum; we derive the corresponding asymptotic expansion in our model. Finally, in Sec.~\ref{sec:dpert} we consider two distant perturbations and show that their mutual influence decays exponentially with their distance in analogy with the Agmon metric effects \cite{Ag82} in the usual theory of Schr\"odinger operators.

%%%%%%%%%%%%%%%%%%%%%%%%%%%%%%%%%%%%%%%%%%%%%%%%%%%%%%%%%%
%%  The periodic case                                   %%
%%%%%%%%%%%%%%%%%%%%%%%%%%%%%%%%%%%%%%%%%%%%%%%%%%%%%%%%%%

\section{The spectrum of the periodic system}
\label{sec:period}

 % -------------- %
\begin{figure}[h!]
\centering
    \includegraphics[scale=0.9]{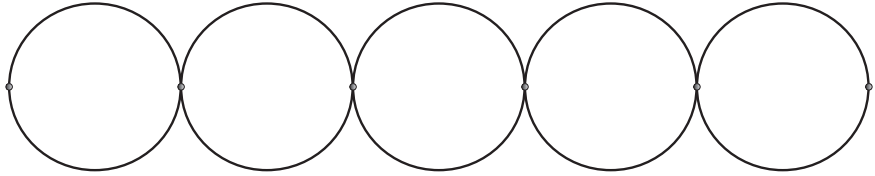}
\caption{The chain graph $\Gamma$}\label{fig:graph}
\end{figure}
 % -------------- %

Let us first consider the ring chain $\Gamma$ as sketched in Fig.~\ref{fig:graph}; without loss of generality we may and will suppose that the circumference of each ring is $2\pi$. We suppose that the particle is moving in the magnetic field generated by the vector potential $\bf A$. The field is assumed to be perpendicular to the graph plane and homogeneous\footnote{In fact, however, the only quantity of importance will be the magnetic flux through the rings, hence in general the field must be just invariant with respect to discrete shifts along the chain.}. The corresponding vector potential can be thus chosen tangential to each ring and constant; since the coordinates we use to parametrize $\Gamma$ refer to different orientations in the upper and lower part of the chain, respectively, we choose $-A$ as the potential value on the upper halfcircles and $A$ on the lower ones.

We denote by $-\Delta_\alpha$ the particle Hamiltonian which acts as $(-\rmi\nabla-{\bf A})^2$ on each graph link, has the domain consisting of all functions from $H^2_\mathrm{loc}(\Gamma)$ which satisfy the boundary conditions (\ref{cond:deltaCoupling1}) at the vertices of $\Gamma$ with the quasiderivatives being equal to the sum of the derivative and the function value multiplied by the tangential component of $e{\bf A}$ as usual \cite{KS03}. In the present case, however, we have two pairs of vectors of opposite orientation so their contributions cancel and the left-hand side of (\ref{cond:deltaCoupling1}) is in fact nothing else than the sum of the derivatives taking into account different coordinate orientations. In this section we suppose that the coupling constant $\alpha$ is the same at each vertex, and we are going to determine the band-and-gap structure of the spectrum.

 % -------------- %
\begin{figure}[h!]
\centering
    \includegraphics[scale=1.2]{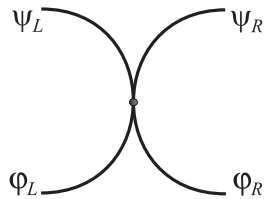}
\caption{Elementary cell of the periodic system}\label{fig:cell}
\end{figure}
 % -------------- %

In view of the periodicity of $\Gamma$ and $-\Delta_\alpha$ with respect to the discrete shifts, the spectrum can be computed using Bloch-Floquet decomposition \cite[Sec.~4.2]{BK13}. Let us consider an elementary cell with the wave function components denoted according to Fig.~\ref{fig:cell} and look for the spectrum of the Floquet components of $-\Delta_\alpha$. Since the operator acts as $-\mathcal{D}^2:=(\rmi\frac{\rmd}{\rmd x}- A)^2$ on the upper halfcircles and as $-\mathcal{D}^2:=(\rmi\frac{\rmd}{\rmd x}+ A)^2$ on the lower ones, each component of the eigenfunction with energy $E:= k^2\neq0$ is a linear combination of the functions $\rme^{-\rmi Ax}\rme^{\pm\rmi kx}$ on the upper graph links and of $\rme^{\rmi Ax}\rme^{\pm\rmi kx}$ on the lower ones. In what follows we conventionally employ the principal branch of the square root, namely, the momentum $k$ should be positive for $E>0$, while for $E$ negative we put $k=\rmi\kappa$ with $\kappa>0$. For a given $E\neq0$, the wave function components on the elementary cell are therefore given by
 % -------------- %
\begin{eqnarray}%\label{wavefunct}
\psi_L(x)&=\rme^{-\rmi Ax}
(
C_L^+\rme^{\rmi k  x}+C_L^-\rme^{-\rmi k  x}
)\,,\qquad
&x\in[-\pi/2,0]\,,
\nonumber\\
\psi_R(x)&=
\rme^{-\rmi Ax}
(
C_R^+\rme^{\rmi k  x}+C_R^-\rme^{-\rmi k  x}
)\,,
&x\in[0,\pi/2]\,,
\nonumber\\
[-.8em] \label{wavefunct}
\\
[-.8em]
\varphi_L(x)&=\rme^{\rmi Ax}
(D_L^+\rme^{\rmi k  x}+D_L^-\rme^{-\rmi k  x})\,,
&x\in[-\pi/2,0]\,,
\nonumber\\
\varphi_R(x)&=
\rme^{\rmi Ax}
(
D_R^+\rme^{\rmi k  x}+D_R^-\rme^{-\rmi k  x}
)\,,
&x\in[0,\pi/2]\,.\nonumber
\end{eqnarray}
 % -------------- %
As we have said, at the contact point the $\delta$-coupling is assumed, i.e. we have
 % -------------- %
\begin{eqnarray}
\nonumber
\phantom{-}
\psi_L(0)=\psi_R(0)=\varphi_L(0)=\varphi_R(0)\,,
\\
[-.8em] \label{deltacoupl}
\\
[-.8em]
-\mathcal{D}\psi_L(0)+\mathcal{D}\psi_R(0)-\mathcal{D}\varphi_L(0)
+\mathcal{D}\varphi_R(0)=\alpha\psi_L(0)\,.
\nonumber
\end{eqnarray}
 % -------------- %
On the other hand, at the `free' ends of the cell the Floquet conditions are imposed,
 % -------------- %
\begin{eqnarray}
\nonumber
\psi_R(\pi/2)=\rme^{\rmi\theta}\psi_L(-\pi/2),
\qquad
\mathcal{D}\psi_R(\pi/2)=\rme^{\rmi\theta}\mathcal{D}\psi_L(-\pi/2)\,,
\\
[-.8em] \label{floquet}
\\
[-.8em]
\varphi_R(\pi/2)=\rme^{\rmi\theta}\varphi_L(-\pi/2),
\qquad
\mathcal{D}\varphi_R(\pi/2)=\rme^{\rmi\theta}\mathcal{D}\varphi_L(-\pi/2)\,,
\nonumber
\end{eqnarray}
 % -------------- %
with $\theta$ running through $[-\pi,\pi)$; alternatively we may say that the quasimomentum $\frac1{2\pi}\theta$ runs through $[-1/2,1/2)$, the Brillouin zone of the problem. In both cases the vector potential contributions subtract and $\mathcal{D}$ can be replaced by the usual derivative.

Substituting \eref{wavefunct} into \eref{deltacoupl} and \eref{floquet}, one obtains after simple manipulations an equation for the phase factor $\rme^{\rmi\theta}$, namely
 % -------------- %
\begin{equation}\label{squareEq}
\sin k\pi\cos A\pi
(\rme^{2\rmi\theta}-2\xi(k)\rme^{\rmi\theta}
+1)=0,
\end{equation}
 % -------------- %
with
 % -------------- %
\[
\xi(k) := \frac{1}{\cos A\pi}
\Big(\cos k\pi + \frac{\alpha}{4k}\sin k\pi\Big)\,,
\]
 % -------------- %
which has real coefficients for any $ k \in\mathbb{R}\cup\rmi\mathbb{R}\setminus\{0\}$ and the discriminant equal to
 % -------------- %
\[
D=4(\xi(k)^2-1)\,.
\]
 % -------------- %
We will treat the special cases $A-\frac12\in\mathbb{Z}$ and $k \in\mathbb{N}$ later, now we will suppose $A-\frac12$ does not belong to $\mathbb{Z}$, the set of integer numbers, while $k$  does not belong to $\mathbb{N}$, the set of natural numbers. We have to determine values of $k ^2$ for which there is a $\theta\in[-\pi,\pi)$ such that \eref{squareEq} is satisfied, in other words, for which $ k ^2$ it has, as an equation in the unknown
$\rme^{\rmi\theta}$, at least one root of modulus one. Note that a pair of solutions of \eref{squareEq} always give one when multiplied, regardless the value of $ k $, hence either both roots are complex conjugated of modulus one, or one is of modulus greater than one and the other has modulus smaller than one. Obviously, the latter situation corresponds to a positive discriminant, and the former one to the discriminant less or equal to zero. We summarize this discussion as follows:

 % -------------- %
\begin{prop}\label{prop:spectr}
Suppose that $A-\frac12\notin\mathbb{Z}$ and $k\in\mathbb{R}^+\cup\rmi\mathbb{R}^+\setminus\mathbb{N}$. Then $k ^2\in\sigma(-\Delta_\alpha)$ holds if and only if the condition
 % -------------- %
\begin{equation}
\label{ineq:spectr}
|\xi(k)|
\leq 1
\end{equation}
 % -------------- %
is satisfied.
\end{prop}
 % -------------- %

\medskip

In particular, the negative spectrum is obtained by putting $ k =\rmi\kappa$ for $\kappa>0$ and rewriting the inequality \eref{ineq:spectr} in terms of this variable. Note that since $\sinh x\neq0$ for all $x > 0$, it never occurs that $\sin k \pi = 0$ for $ k \in\rmi\mathbb{R}^+$, the positive imaginary axis, thus there is no need to treat this case separately like for $ k \in \mathbb{R}^+$, cf. \eref{squareEq} above.

 % -------------- %
\begin{cor}
If $A-\frac12\notin\mathbb{Z}$ and $\kappa>0$, then $-\kappa^2\in\sigma(-\Delta_\alpha)$ holds if and only if
 % -------------- %
\begin{equation*}
\Big|
\frac{1}{\cos A\pi}
\Big(
\cosh\kappa\pi +\frac{\alpha}{4\kappa}
\sinh\kappa\pi
\Big)
\Big|
\leq 1\,.
\end{equation*}
 % -------------- %
\end{cor}
 % -------------- %

\medskip

\noindent Observe that in the case $E=0$ we repeat similar arguments to get the equation
 % -------------- %
\begin{equation*}
\rme^{2\rmi\theta}-
\frac{2\rme^{\rmi\theta}}{\cos A\pi}\Big(1+\frac{\alpha\pi}{4}\Big)+1=0\,,
\end{equation*}
 % -------------- %
replacing \eref{squareEq}, whence we infer that $0\in\sigma(-\Delta_\alpha)$ holds if and only if
 % -------------- %
\[
\Big|\frac{1}{\cos A\pi}\Big(1+\frac{\alpha\pi}{4}\Big)\Big|
\leq 1\,,
\]
 % -------------- %
hence zero can belong to the continuous part of the spectrum only and this happens \emph{iff} $\alpha\in[-4(|\cos A\pi|+1)/\pi,4(|\cos A\pi|-1)/\pi]$.

Let us finally mention the cases $k\in\mathbb{N}$ and  $A-\frac12\in\mathbb{Z}$ left out above. In the former one it is straightforward to check that $k^2$ is an eigenvalue, and moreover, that it has an infinite multiplicity. One can construct an eigenfunction which is supported by two adjacent circles, which is given by $\psi_L(x) =\rme^{-\rmi Ax} \sin kx$ and $\psi_R(x) =(-1)^{k+1}\rme^{\rmi A(\pi-x)} \sin kx$ with $x\in[0,\pi]$ on the upper semicircles and $\psi_L(x) = -\rme^{\rmi Ax}\sin kx$  and $\psi_R(x) =(-1)^ k \rme^{\rmi A(x-\pi)}\sin kx$ with $x\in[0,\pi]$ on the lower ones. In the extremal case $A\in\mathbb{Z}$ when $|\cos A\pi|=1$, we can construct an eigenfunction supported by a single circle, namely, $\psi(x) =\rme^{-\rmi Ax} \sin kx$ on the upper semicircle and $\psi(x) = -\rme^{\rmi Ax}\sin kx$ on the lower one. In the case $A-\frac12\in \mathbb{Z}$ the spectrum is produced by solutions to the equation $\cos k\pi +\frac{\alpha}{4k} \sin k\pi=0$, each of which yields an eigenvalue of an infinite multiplicity.
One can construct an eigenfunction supported by two adjacent circles, which is given by $\psi_L(x)=-\rme^{-\rmi Ax}\sin kx$, $x\in[0,\pi]$, and
$\psi_R(x)=\rme^{\rmi A(\pi-x)}\sin k(\pi-x)$, $x\in[0,\pi]$, on the upper semicircles and by $\varphi_L(x)=\rme^{\rmi Ax}\sin kx$, $x\in[0,\pi]$,
and $\varphi_R(x)=\rme^{\rmi A(x-\pi)}\sin k(x-\pi)$, $x\in[0,\pi]$, on the lower ones.

In conclusion, we can make the following claim about $\sigma(-\Delta_\alpha)$.

 % -------------- %
\begin{thm}[Magnetic case]
\label{thm:spectrumUnperturbed}
Assume that $A\notin\mathbb{Z}$. If $A-\frac12\in\mathbb{Z}$; then the spectrum of $-\Delta_\alpha$ consists of two series of infinitely degenerate eigenvalues $\{k^2\in\mathbb{R}\colon\, \xi(k)=0\}$ and $\{k^2\in\mathbb{R}\colon\, k\in\mathbb{N}\}$.

On the other hand, suppose that $A-\frac12\notin\mathbb{Z}$; then the spectrum of $-\Delta_\alpha$ consists of infinitely degenerate eigenvalues equal to $k^2$ with $k\in\mathbb{N}$, and absolutely continuous spectral bands with the following properties:

Every spectral band except the first one is contained in an interval $(n^2,(n+1)^2)$ with $n\in\mathbb{N}$. The position of the first spectral band depends on $\alpha$, namely, it is included in $(0,1)$ if  $\alpha>4(|\cos A\pi|-1)/\pi$ or it is negative if $\alpha<-4(|\cos A\pi|+1)/\pi$, otherwise the first spectral band contains zero.
\end{thm}
 % -------------- %

\medskip

\noindent \textsl{The proof} follows directly from Proposition~\ref{prop:spectr} and the explicit formul{\ae} given above. The behavior of the first spectral band as a function of the coupling constant for a fixed value of the magnetic flux is illustrated in Fig.~\ref{fig:FirstBand}.

 % -------------- %
\begin{figure}[h!]
\centering
    \includegraphics[scale=0.5]{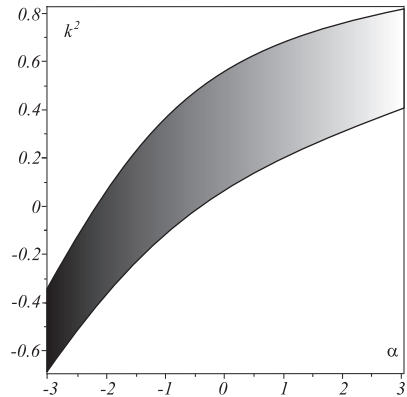}
  \caption{The first spectral band of the operator $-\Delta_\alpha$ against $\alpha$ at $\cos A\pi=0.7$.}\label{fig:FirstBand}
\end{figure}
 % -------------- %

The reader may wonder why integer values of $A$ are not included in the `magnetic' case. The reason is seen from the fact that the above spectral condition is invariant with respect to the change of $A$ by an integer which reflects the existence of a simple gauge transformation between such cases. We note that in the chosen units the magnetic flux quantum is $2\pi$ and $A=\frac12 B= \frac{1}{2\pi}\Phi$; we can then rephrase the above claim saying the systems differing by an integer number of flux quanta through each ring are physically equivalent. In this sense the case of an integer $A$ is  thus equivalent to the non-magnetic chain treated in \cite{DET08}; to make the exposition self-contained, it is nevertheless useful to single it out and state it explicitly.

 % -------------- %
\begin{thm}[Non-magnetic case]
\label{thm:spectrumUnperturbed}
Suppose that $A\in\mathbb{Z}$; then the spectrum of $-\Delta_\alpha$ consists of infinitely degenerate eigenvalues equal to $k^2$ with $k\in\mathbb{N}$, and absolutely continuous spectral bands with the following properties:

If $\alpha > 0$, then every spectral band is contained in an interval $(n^2,(n+1)^2]$ with $n\in\mathbb{N}\cup\{0\}$, and its upper edge coincides with the value $(n+1)^2$.

If $\alpha<0$, then in each interval $[n^2,(n+1)^2)$ with $n\in\mathbb{N}$ there is exactly one spectral band the lower edge of which coincides with $n^2$. In addition, there is a spectral band with the lower edge equal to $k^2<0$, where $k$ is the solution to the equation $|\xi(k)|=1$ with the smallest square. The upper edge of this band is produced by the solution of the equation with the second smallest square. The position of the upper edge of this band depends on $\alpha$, namely if $-8/\pi<\alpha<0$, then it is contained in $(0,1)$. On the other hand, for $\alpha<-8/\pi$ the upper edge is negative and for $\alpha=-8/\pi$ it equals zero.
\end{thm}
 % -------------- %

\medskip

It is also worth mentioning that if $|\cos A\pi|<1$, the interval $(n^2,(n+1)^2)$ contains one band and two gaps parts. Moreover, if $|\cos A\pi|$ increases, the both gaps shrink, and in addition, if  $\alpha$ is positive (negative), then the right (respectively, left) gap vanishes in the limit  $|\cos A\pi|\to1$. Finally, for $\alpha=0$ the spectrum of the Hamiltonian coincides with the positive half-line for $|\cos A\pi|=1$, while it keeps the band structure with open gaps for $|\cos A\pi|<1$.

%%%%%%%%%%%%%%%%%%%%%%%%%%%%%%%%%%%%%%%%%%%%%%%%%%%%%%%%%%
%%  The duality                                         %%
%%%%%%%%%%%%%%%%%%%%%%%%%%%%%%%%%%%%%%%%%%%%%%%%%%%%%%%%%%

\section{Duality between differential and difference operators}
\label{sec:duality}

Denote by $\balpha=\{\alpha_j\}_{j\in\mathbb{Z}}$ an arbitrary but fixed sequence of real numbers and consider the corresponding magnetic Laplacian $-\Delta_{\balpha}$ on our chain graph, that is, the operator acting as $-\mathcal{D}^2$ on each graph edge with the domain consisting of those functions from the Sobolev space that satisfy the $\delta$ boundary conditions at the graph vertices,
 % -------------- %
\begin{eqnarray}
\label{discOperCoupl1}
\psi_j(j\pi) = \varphi_j(j\pi) = \psi_{j-1}(j\pi) = \varphi_{j-1}(j\pi)\,,
\\
\label{discOperCoupl2}
\mathcal{D}\psi_j(j\pi)+\mathcal{D}\varphi_j(j\pi)-
\mathcal{D}\psi_{j-1}(j\pi)-\mathcal{D}\varphi_{j-1}(j\pi)
=\alpha_j\psi_j(j\pi)\,.
\end{eqnarray}
 % -------------- %
As before the wave function component $\psi_j$ corresponds here to the upper halfcircle of the $j$th ring while $\varphi_j$ stands for the lower one, and $\mathcal{D}$ can be replaced by the usual derivative. Our aim in this section is to formulate a one-to-one correspondence between the differential  operator $-\Delta_{\balpha}$ in $L^2(\Gamma)$ and a certain operator acting in $\ell_2(\mathbb{Z})$. We seek a difference equation such that every bounded (square summable) solution of it gives rise to a bounded (square integrable) solution of the Schr\"{o}dinger equation corresponding to $-\Delta_{\balpha}$, and vice versa every bounded (square integrable) solution of the Schr\"{o}dinger equation produces a bounded (square summable) solution of the difference equation. This connection will play a crucial role in the following sections in determining the spectrum of $-\Delta_{\balpha}$ for particular sequences $\balpha$.

Denote by $\psi\choose\varphi$ the general solutions of the Schr\"{o}dinger equation
 % -------------- %
\[
(-\Delta_{\balpha}- k^2)
\left(
	\begin{array}{c}
	\psi(x,k)\\[.5em] \varphi(x,k)
	\end{array}
\right)=0\,,\qquad \Im  k\geq 0\,,
\]
 % -------------- %
which can be rewritten componentwise on the upper and lower semicircles as follows
 % -------------- %
\begin{eqnarray}\nonumber
(-\mathcal{D}^2- k^2)\psi(x, k)&=&0
\,,\qquad \Im  k\geq 0\,,\quad x\in \mathcal{I}_j\,,
\\[-.8em]
\label{eq:psi&phi}
\\[-.8em]\nonumber
(-\mathcal{D}^2- k^2)\varphi(x, k)&=&0
\,,\qquad \Im  k\geq 0\,,\quad x\in \mathcal{I}_j\,,
\end{eqnarray}
 % -------------- %
where the interval $\mathcal{I}_j:=(j\pi,(j+1)\pi)$ corresponds to both the upper and lower halfcircles. Recalling that the wave function should satisfy the continuity  condition \eref{discOperCoupl1}, we see that the general solutions $\psi$ and $\varphi$ and their quasiderivatives are given by
 % -------------- %
\begin{eqnarray}\nonumber
\fl
\psi(x, k)=
\rme^{\rmi A(j\pi-x)}
\bigg(
\psi(j\pi, k)\cos k(x-j\pi)+\mathcal{D}\psi((j\pi)^+, k)\,\frac{\sin  k(x-j\pi)}{ k}
\bigg)
\,,\\[-.8em]
\label{eq:psi}
\\[-0.8em]\nonumber
\fl\mathcal{D}\psi(x, k)=
\rme^{\rmi A(j\pi-x)}
(
-\psi(j\pi, k)\, k\sin k(x-j\pi)+\mathcal{D}\psi((j\pi)^+, k)\cos k(x-j\pi)
)
\end{eqnarray}
 % -------------- %
and
 % -------------- %
\begin{eqnarray}\nonumber
\fl \varphi(x, k)=
\rme^{\rmi A(x-j\pi)}
\bigg(
\varphi(j\pi, k)\cos k(x-j\pi)+\mathcal{D}\varphi((j\pi)^+, k)
\,\frac{\sin  k(x-j\pi)}{ k}
\bigg)
\,,\\[-.8em]
\label{eq:phi}
\\[-.5em]\nonumber
\fl\mathcal{D}\varphi(x, k)=
\rme^{\rmi A(x-j\pi)}
(
-\varphi(j\pi, k)\,  k\sin k(x-j\pi)+\mathcal{D}\varphi((j\pi)^+, k)\cos  k(x-j\pi))
\end{eqnarray}
 % -------------- %
for $\Im k\geq0$ and $x\in\mathcal{I}_j$. Next, let us introduce the vector
 % -------------- %
\[
\Psi_j(k,\tau)=
\left(
\begin{array}{c}
\rme^{\tau}\psi(j\pi, k)+\rme^{-\tau}\varphi(j\pi,k)
\\[0.5em]
\rme^{\tau}\mathcal{D}\psi((j\pi)^-, k)+\rme^{-\tau}\mathcal{D}\varphi((j\pi)^-, k)
\end{array}\right),
\]
 % -------------- %
and the matrix
 % -------------- %
\[
S_j(k)=
\left(\begin{array}{cc}
\cos k\pi+\frac{\alpha_j}{2 k}\sin k\pi&
\frac1 k\sin k\pi
\\[0.5em]
- k\sin k\pi +\frac{\alpha_j}2\cos k\pi
&
\cos k\pi
\end{array}\right)\,,
\]
 % -------------- %
then taking into account relation \eref{discOperCoupl2} we conclude that
 % -------------- %
\[
S_j( k)\Psi_j(k,0)
=
\Psi_{j+1}(k,\rmi A\pi)\,,
\qquad \Im  k \geq 0,\quad j\in\mathbb{Z}\,.
\]
 % -------------- %
In a similar spirit, we introduce the vector
 % -------------- %
\[
\Phi_j( k)=
\left(
\begin{array}{c}
\psi(j\pi, k)+\varphi(j\pi, k)
\\[0.5em]
\psi((j-1)\pi,k)+\varphi((j-1)\pi,k)
\end{array}\right)
\]
 % -------------- %
to obtain, in view of \eref{eq:psi} and \eref{eq:phi}, the relation
 % -------------- %
\[
(\cos A\pi)\Phi_j( k)
=
T( k)\Psi_j( k,0)
\]
 % -------------- %
where the matrix $T$ is defined as follows
 % -------------- %
\[
T( k)
=
\left(\begin{array}{cc}
\cos A\pi & 0
\\[0.5em]
\cos k\pi & -\frac1 k\sin k\pi
\end{array}\right)\,.
\]
 % -------------- %
Finally, define the matrix
 % -------------- %
\[
N_j( k)
=T(k)S_j(k)(T(k))^{-1}
=\left(\begin{array}{cc}
2\xi_j(k)&-1
\\[.5em]
1&0
\end{array}\right)\,,
\]
 % -------------- %
where $ k\in\mathfrak{K}:=\{z\colon\Im z\geq0\wedge z\notin\mathbb{Z}\}$ and
 % -------------- %
\[
\xi_j(k)=\frac{1}{\cos A\pi}
\Big(
	\cos k\pi
		+
	\frac{\alpha_j}{4k}\sin k\pi
\Big)\,,
\]
 % -------------- %
to get the relation
 % -------------- %
\[
{N}_j( k)\Phi_j( k)=\Phi_{j+1}( k)\,,\qquad  k\in\mathfrak{K}\,,
\]
 % -------------- %
which by continuity of the wave function at the graph vertices can be rewritten as
 % -------------- %
\[
{N}_j( k)
\left(
	\begin{array}{c}
	\psi_j(k)\\[.5em]
	\psi_{j-1}(k)
	\end{array}
\right)
=
\left(
	\begin{array}{c}
	\psi_{j+1}(k)\\[.5em]
	 \psi_{j}(k)
	\end{array}
\right)
\,,\qquad  k\in\mathfrak{K}\,,
\]
 % -------------- %
or equivalently as
 % -------------- %
\begin{equation}\label{discr_relat}
\psi_{j+1}(k)
+
\psi_{j-1}(k)
=
\xi_j(k)\psi_j(k)\,,\qquad  k\in\mathfrak{K}\,,
\end{equation}
 % -------------- %
where $\psi_j( k):=\psi(j\pi, k)$. Summing up the above reasoning, we have arrived at the following conclusion:

 % -------------- %
\begin{thm}\label{thm:Duality}
Suppose that $\alpha_j\in\mathbb{R}$; then any solutions $\psi(\cdot, k)$ and $\varphi(\cdot, k)$, $ k^2\in\mathbb{R}$, $ k\in\mathfrak{K}$, of \eref{eq:psi&phi} satisfy relation \eref{discr_relat}. Conversely, any solution of the difference equation \eref{discr_relat} defines via
 % -------------- %
\begin{eqnarray}\nonumber
\psi(x, k)&=
\rme^{-\rmi A(x-j\pi)}
\bigg[
\psi_j(k)\cos k(x-j\pi)
\\\nonumber
&\phantom{=}\,
+
(\psi_{j+1}( k)\rme^{\rmi A\pi}-\psi_j(k)\cos k\pi)
\frac{\sin k(x-j\pi)}{\sin k\pi}
\bigg]
\,,\quad x\in\mathcal{I}_j\,,
\\[-.6em]
\label{eq:psi(x)VSpsi(k)}
\\[-.6em]\nonumber
\varphi(x, k)&=
\rme^{\rmi A(x-j\pi)}
\bigg[
\psi_j(k)\cos k(x-j\pi),
\\\nonumber
&\phantom{=}\,
+
(\psi_{j+1}(k)\rme^{-\rmi A\pi}-\psi_j(k)\cos k\pi)
\frac{\sin k(x-j\pi)}{\sin k\pi}
\bigg]
\,,
\quad x\in\mathcal{I}_j\,,
\end{eqnarray}
 % -------------- %
solutions of equations \eref{eq:psi&phi} satisfying $\delta$-coupling conditions \eref{discOperCoupl1}, \eref{discOperCoupl2}. In addition,
$\Big(\begin{array}{c}\psi(\cdot, k)\\[-.3em]\varphi(\cdot, k)\end{array}\Big)\in L^p(\Gamma)$ if and only if $\{\psi_j(k)\}_{j\in\mathbb{Z}}\in \ell_p(\mathbb{Z})$ for $p\in\{2,\infty\}$.
\end{thm}
 % -------------- %
\begin{proof}
It remains to demonstrate the last statements. Let $k^2\in\mathbb{R}$, $k\in\mathfrak{K}$, and assume that all the solutions $\psi(\cdot, k)$, $\varphi(\cdot, k)$, and $\psi_j( k)$ are real-valued. If $\psi,\varphi\in L^p(\mathbb{R})$, and thus $\mathcal{D}^2\psi,\mathcal{D}^2\varphi\in L^p(\mathbb{R})$, we infer that $\mathcal{D}\psi,\mathcal{D}\varphi\in L^p(\mathbb{R})$ holds for all $1\leq p\leq\infty$. Then $\{\psi_j(k)\}_{j\in\mathbb{Z}}\in \ell_p(\mathbb{Z})$ follows from
 % -------------- %
\[\fl
\psi(j\pi, k)=\rme^{\rmi A(x-j\pi)}
\bigg(
\psi(x,k)\cos k(x-j\pi)-\mathcal{D}\psi(x, k)\,
\frac{\sin k(x-j\pi)}k
\bigg)
\,,\qquad x\in\mathcal{I}_j\,,
\]
 % -------------- %
for $p=\infty$ and from
 % -------------- %
\begin{eqnarray*}
&(\psi(j\pi,k))^2+\frac1{k^2}(\mathcal{D}\psi((j\pi)^+,k))^2
\\
&\qquad
=\rme^{2\rmi A(x-j\pi)}
\bigg(
(\psi(x,k))^2+\frac1{k^2}
(\mathcal{D}\psi(x,k))^2
\bigg)\,,\qquad x\in\mathcal{I}_j\,,
\end{eqnarray*}
 % -------------- %
for $p=2$. Conversely, assume $\{\psi_j( k)\}_{j\in\mathbb{Z}}\in \ell_p(\mathbb{Z})$ for $p=\infty$ or $p=2$. The case $p=\infty$ directly results from \eref{eq:psi(x)VSpsi(k)} and the case $p=2$ follows from \eref{eq:psi(x)VSpsi(k)} and
 % -------------- %
\begin{eqnarray*}
&\rme^{2\rmi A(x-j\pi)}
\bigg(
(\psi(x,k))^2+\frac1{k^2}
(\mathcal{D}\psi(x,k))^2
\bigg)
\\
&\qquad
=
(\psi_j(k))^2+
    \bigg(
        \frac{\psi_{j+1}(k)\rme^{\rmi A\pi}-\psi_j(k)\cos k\pi}{\sin k\pi}
    \bigg)^2
\,,\qquad x\in\mathcal{I}_j\,,
\end{eqnarray*}
 % -------------- %
which completes the proof of the theorem.
\end{proof}

%%%%%%%%%%%%%%%%%%%%%%%%%%%%%%%%%%%%%%%%%%%%%%%%%%%%%%%%%%
%%  Local impurities                                    %%
%%%%%%%%%%%%%%%%%%%%%%%%%%%%%%%%%%%%%%%%%%%%%%%%%%%%%%%%%%

\section{Systems with local impurities}
\label{sec:mcoupl}

Let the Hamiltonian of the periodic system $-\Delta_\alpha$ be defined as in previous sections, and choose $\{\pi,2\pi,\ldots,m\pi\}$ as the set of coordinate values at which the vertex coupling suffers a perturbation that models an impurity changing the $\delta$-interaction coupling strength
by $\gamma_j$, $j\in\mathbb{M}:=\{1,2,\ldots,m\}$. The perturbed Hamiltonian will be then denoted by $-\Delta_{\balpha+\bgamma}$. Our goal here is to relate spectral properties of the two operators. Since the resolvents of $-\Delta_{\balpha+\bgamma}$ and $-\Delta_\alpha$ differ by a finite-rank operator the essential spectra of the two operators coincide, of course. However, the perturbation may give rise to eigenvalues in the gaps and we are going to derive a characteristic equation giving this discrete spectrum. It will contain a determinant given in terms of a recurrent expression, which makes its analysis in the general case a rather challenging task. Because of this complexity, we shall analyze in detail a few particular cases, starting with $m=1$ and $m=2$, followed by perturbation with an arbitrary natural $m$ but with all the impurities identical, in other words, $\gamma_1=\ldots=\gamma_m$.

As we have shown in the previous section, wave functions in the neighboring rings are related by the matrix $N_j$ as follows,
 % -------------- %
\[
\Phi_{j+1}(k)
=
N_j(k)
\Phi_j(k)
\,,\qquad j\in\mathbb{Z}\,.
\]
 % -------------- %
Since these matrices are the same outside the perturbation support, we have
 % -------------- %
\begin{eqnarray}
        \label{cond:efViaPhiM+1}
\Phi_{m+j+1}(k) &= (N(k))^j\Phi_{m+1}(k)\,,
    \qquad &j \in\mathbb{N}\,, \\
        \label{cond:efViaPhiJ}
\Phi_{j+1}(k) &= N_j(k)\Phi_{j}(k)\,,
    \qquad &j \in\mathbb{M}\,, \\
        \label{cond:efViaPhi1}
\Phi_j(k)     &= (N(k))^{j-1}\Phi_1(k)\,,
    \qquad -&j\in\mathbb{N}_0\,,
\end{eqnarray}
 % -------------- %
where the matrix $N$ corresponds to the function
 % -------------- %
\[
\xi(k)=\frac{1}{\cos A\pi}
\Big(
	\cos k\pi
		+
	\frac{\alpha}{4k}\sin k\pi
\Big)
\]
 % -------------- %
in the same way as the matrix $N_j$ is constructed using $\xi_j$. It is obvious that the asymptotical behavior of the norms of $\Phi_j$ is determined by the spectral properties of the matrix $N$. Specifically, let $\Phi_{m+1}$ be an eigenvector of $N$ corresponding to an eigenvalue $\mu$, then $|\mu| < 1$ ($|\mu|>1$, $|\mu|=1$) means that the norm of $\Phi_j$ decays exponentially with respect to $j>m$ (respectively, it is exponentially growing, or independent of $j$). At the same time if $\Phi_1$ is an eigenvector of $N$ corresponding to an eigenvalue $\mu$ such that $|\mu| > 1$, then the norm of $\Phi_j$ decays exponentially with respect to negative $j$ (with a similar conclusions for $|\mu|<1$ and $|\mu|=1$).

By virtue of Theorem~\ref{thm:Duality} the wave function components on the $j$-th ring are determined by $\Phi_j$, and thus, in view of \eref{cond:efViaPhiM+1} and \eref{cond:efViaPhi1}, by $\Phi_{m+1}$ or $\Phi_1$ depending on the sign of $j$. If $\Phi_{m+1}$ has a non-vanishing component related to an eigenvalue of $N$ of modulus larger than one, or $\Phi_1$ has a non-vanishing component related to an eigenvalue of modulus less than one, then the corresponding coefficients $\Phi_j$ determine neither an eigenfunction nor a generalized eigenfunction of $-\Delta_{\balpha+\bgamma}$. On the other hand, if $\Phi_{m+1}$ is an eigenvector, or a linear combination of eigenvectors, of the matrix $N$ with modulus less than one (respectively, equal to one), and at the same time $\Phi_1$ is an eigenvector, or a linear combination of eigenvectors, of the matrix $N$ with modulus larger than one (respectively, equal to one), then the coefficients $\Phi_j$ determine an eigenfunction (respectively, a generalized eigenfunction) and the corresponding energy $E$ belongs to the point (respectively, continuous) spectrum of the operator $-\Delta_{\balpha+\bgamma}$. To perform the spectral analysis of $N(k)$, we employ its characteristic polynomial at energy $k^2$,
 % -------------- %
\begin{equation}\label{charPoly}
\lambda^2-2\xi(k)\lambda+1\,;
\end{equation}
 % -------------- %
it shows that $N(k)$ has an eigenvalue of modulus less than one \emph{iff} the discriminant of \eref{charPoly} is positive, i.e.
 % -------------- %
\[
|\xi(k)|>1\,,
\]
 % -------------- %
and a pair of complex conjugated eigenvalues of modulus one \emph{iff} the above quantity is less or equal to one. In the former case the eigenvalues of $N(k)$ are given by
 % -------------- %
\begin{equation}\label{eq:eigenvalues}
\lambda_{1,2}(k)=
\xi(k)
\pm\sqrt{\xi(k)^2-1}\,,
\end{equation}
 % -------------- %
satisfying $\lambda_2=\lambda_1^{-1}$, hence $|\lambda_2|<1$ holds if $\xi(k)>1$ and $|\lambda_1|<1$ if this quantity is less than $-1$. Moreover, the corresponding eigenvectors of $N(k)$ are
 % -------------- %
\begin{equation}\label{eigenFunc}
    u_{1,2}(k) = \left(
    \begin{array}{c}1\\[.5em] \lambda_{2,1}(k)\end{array}\right)\,.
\end{equation}
 % -------------- %
It is convenient to define the function
 % -------------- %
\[
\lambda(k)=
\xi(k)-\sgn(\xi(k))\sqrt{\xi(k)^2-1}\,,
\]
 % -------------- %
which coincides with $\lambda_1(k)$ if $\xi(k)<-1$ or with $\lambda_2(k)$ if $\xi(k)>1$, hence the modulus of $\lambda(k)$ is smaller than one unless $k^2\in\sigma(-\Delta_{\alpha})$. Abbreviating
 % -------------- %
\begin{eqnarray*}
\hspace{-3em} P_0(k)&=1\,,\qquad P_1(k)&=2\xi_1(k)\,,
\qquad P_m(k)=2\xi_m(k)P_{m-1}(k)-P_{m-2}(k)\,,
\\
\hspace{-3em} Q_0(k)&=0\,,\qquad Q_1(k)&=1\,,\qquad\qquad  Q_m(k)=2\xi_m(k)Q_{m-1}(k)-Q_{m-2}(k)\,,
\end{eqnarray*}
we arrive at the following conclusion.

 % -------------- %
\begin{prop}\label{pro:EVCondition}
Assume that $k^2\in\mathbb{R}\setminus\sigma(-\Delta_\alpha)$; then $k^2$ is an eigenvalue of $-\Delta_{\balpha+\bgamma}$ \emph{iff} for this $k$ we have
 % -------------- %
\begin{equation}\label{cond:spectrumXi}
Q_{m-1}(k)\lambda(k)^2-
(P_{m-1}(k)+Q_{m}(k))
\lambda(k)+P_{m}(k)
=0\,.
\end{equation}
 % -------------- %
\end{prop}
 % -------------- %
\begin{proof}
Suppose first that $\xi(k)<-1$, so that $|\lambda_1(k)|<1$ and $|\lambda_2(k)|>1$; then by virtue of Theorem~\ref{thm:Duality} and formul{\ae}~\eref{cond:efViaPhiM+1}, \eref{cond:efViaPhi1} the only possibility to construct an eigenfunction of $-\Delta_{\balpha+\bgamma}$ is by demanding that
 % -------------- %
\[
\Phi_{m+1}\sim
u_1\,,\qquad
\Phi_1\sim
u_2\,.
\]
 % -------------- %
Thus condition \eref{cond:efViaPhiJ} implies
 % -------------- %
\begin{equation}\label{cond:spectrumXi<-1New}
\det\{N_m(k)\ldots N_1(k)u_2(k),u_1(k)\}=0\,.
\end{equation}
 % -------------- %
Observe that the product
 % -------------- %
\[
\mathcal{N}_m(k):=\prod_{j=1}^mN_j(k)\,,\quad m\in\mathbb{N}\,,
\]
 % -------------- %
can be recursively expressed in terms of quasi-polynomials $P_m$ and $Q_m$,
 % -------------- %
\[
\mathcal{N}_m(k)=
\left(\begin{array}{cc}
P_m(k)&-Q_m(k)
\\[.5em]
P_{m-1}(k)&-Q_{m-1}(k)
\end{array}\right)
\,,
\quad m\in\mathbb{N}\,.
\]
 % -------------- %
It is worth mentioning here that $\det N_j=1$, hence we have also $\det\mathcal{N}_m=1$. This fact in combination with formula \eref{cond:spectrumXi<-1New} implies the spectral condition \eref{cond:spectrumXi}. The case $\xi(k)>1$ may be handled in the same manner.
\end{proof}

%%%%%%%%%%%%%%%%%%%%%%%%%%%%%%%%%%%%%%
\subsection{Example: a single impurity}

As indicated, we will discuss in detail a few examples. The simplest one concerns the situation with a single impurity, $m=1$. In this case $\mathcal{N}_1(k)=N_1(k)$, and the characteristic equation of Proposition~\ref{pro:EVCondition}, which is a quadratic polynomial with respect to $\lambda$, is reduced to a linear equation giving thus the following condition for eigenvalue existence
 % -------------- %
\begin{equation}\label{m=1}
\gamma_1=-\sgn(\xi(k))\frac{4k\cos A\pi}{\sin k\pi}
\sqrt{\xi(k)^2-1}\,.
\end{equation}
 % -------------- %
Let $f$ stand for the right-hand side of relation~\eref{m=1} and suppose that $A\notin\mathbb{Z}$. Then, as $k^2$ varies from the lower end of a gap in $\sigma(-\Delta_\alpha)$ to the upper end, $f(k)$ is continuous with respect to $k$ and strictly increasing with respect to $k^2$. In particular, $f(k)$ alternately increases from $-\infty$ to zero or from zero to $\infty$, starting with the increase from $-\infty$ to zero in the first gap (the one below the continuum spectrum threshold). This behavior is illustrated in Fig.~\ref{fig:m=1}, in view of \eref{m=1} the plots also show the dependence of the eigenvalues in the gaps on the perturbation parameter $\gamma_1$.

 % -------------- %
\begin{figure}[h!]
\centering
    \includegraphics[scale=0.33]{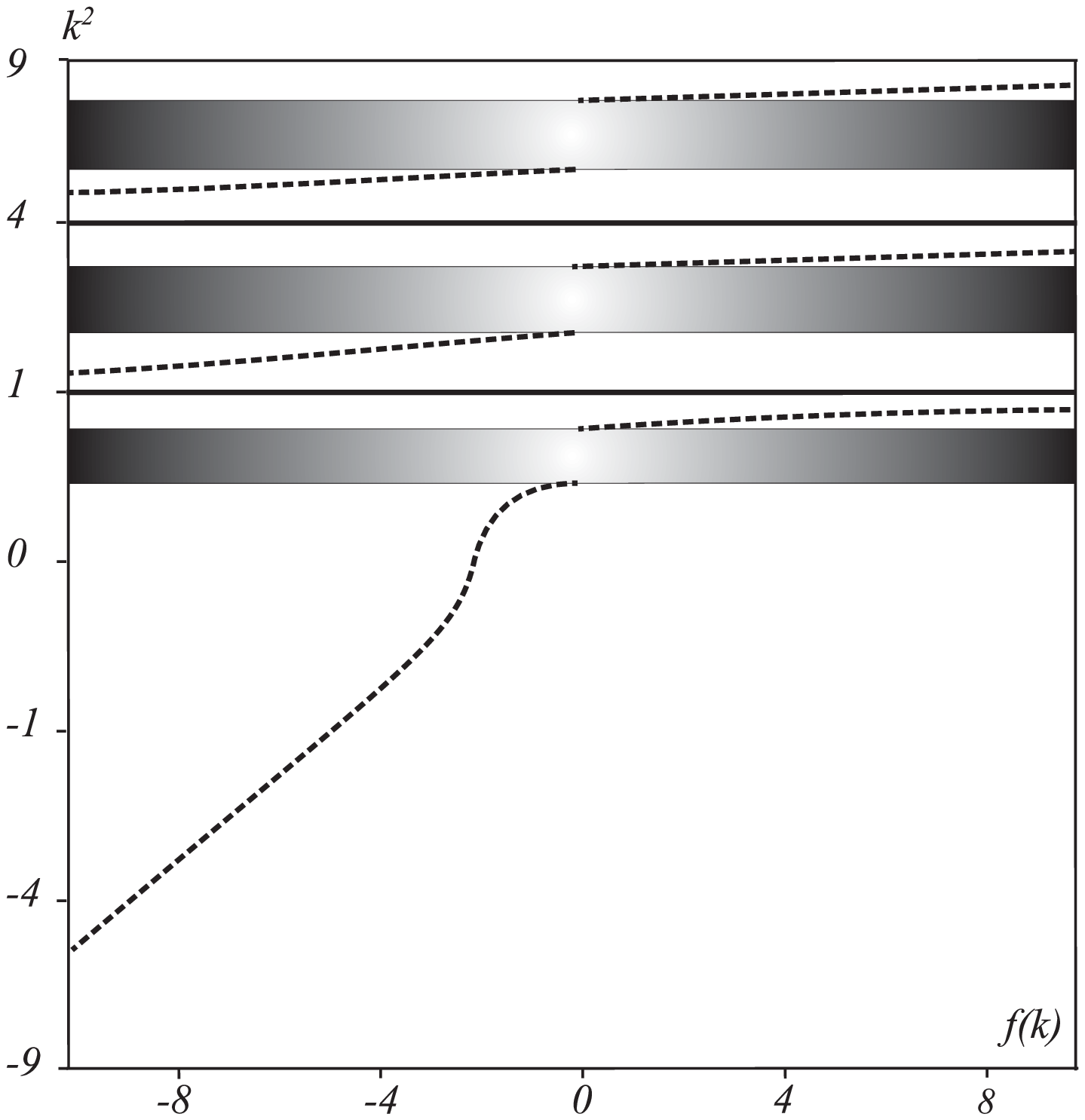}
    \includegraphics[scale=0.33]{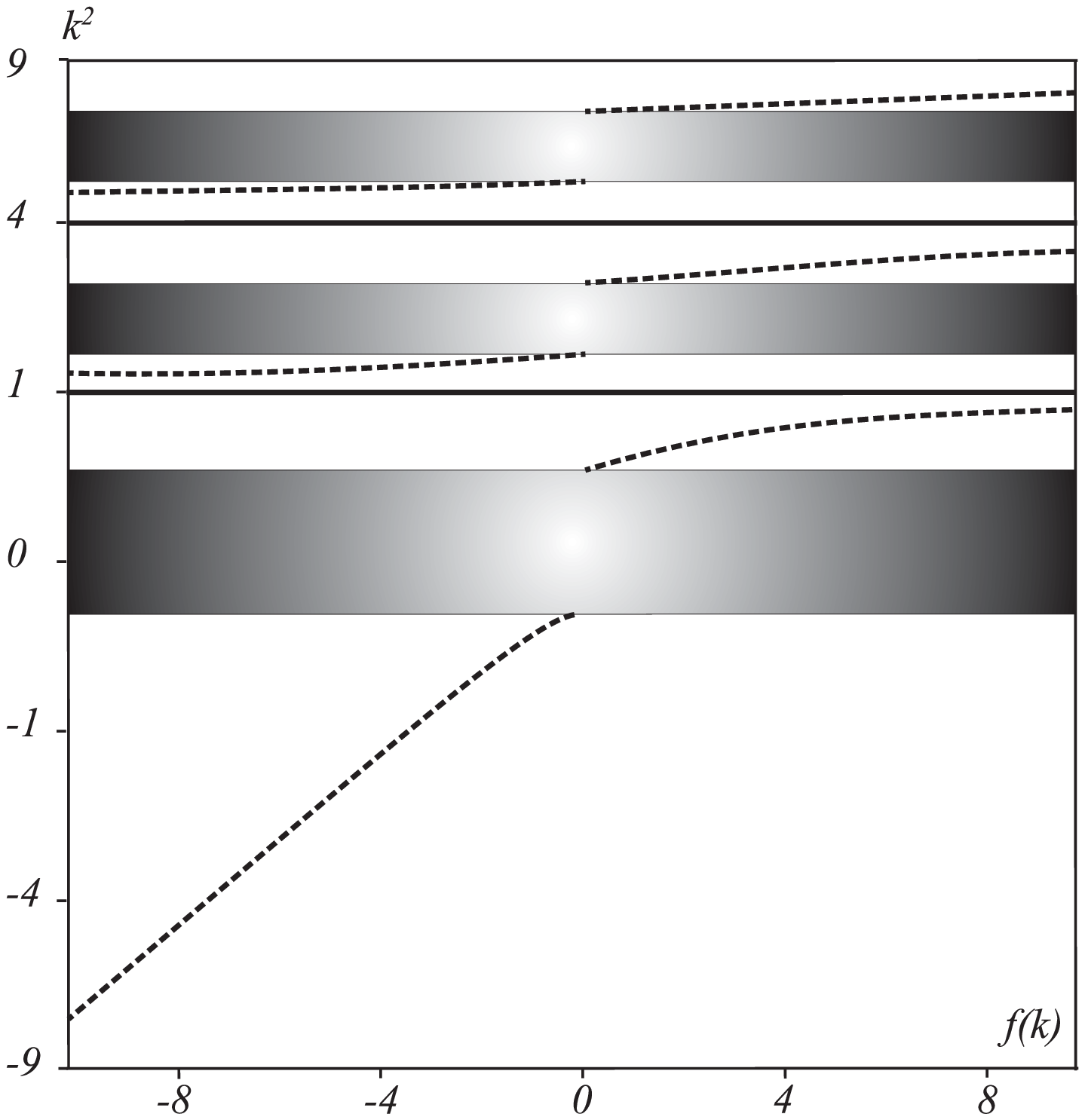}
    \includegraphics[scale=0.33]{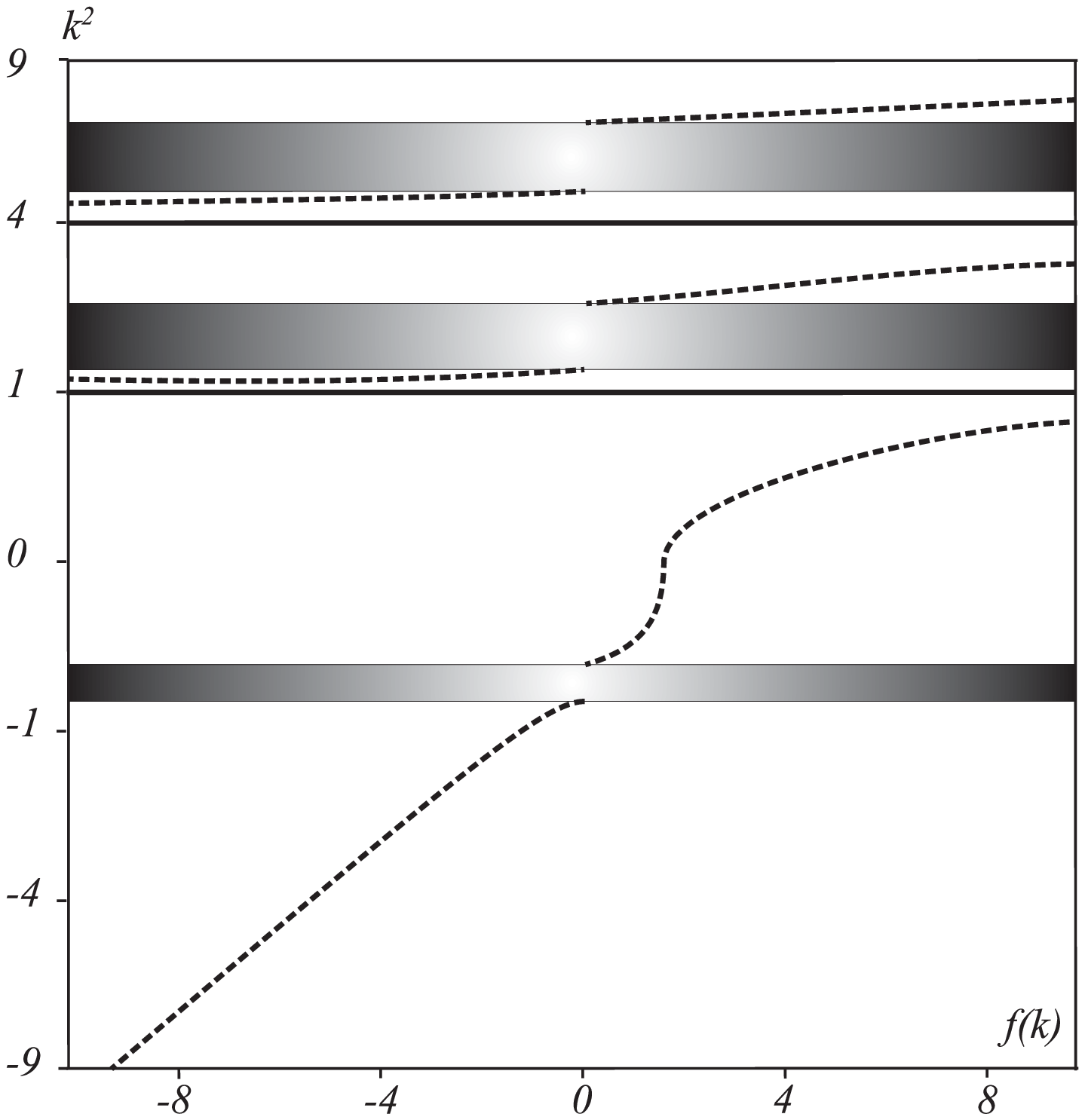}
  \caption{The dependence of $k^2$ on $f(k)$ for $\cos A\pi=0.6$ and the perturbation coupling strength (i) $\alpha=1$, (ii) $\alpha=-1$, (iii) $\alpha=-3$.}\label{fig:m=1}
\end{figure}
 % -------------- %

\noindent Suppose next that $A$ approaches an integer number and $\alpha$ is positive. As we have said at the end of Sec.~\ref{sec:period}, the spectral gaps in the positive spectrum, where $f(k)$ is positive, disappear in the limit. Thus if $A\in\mathbb{Z}$ and $\alpha>0$, then $f(k)$ increases from $-\infty$ to zero in every spectral gap. Likewise, if $A$ approximates an integer and $\alpha<0$, then we should `exclude' those  spectral gaps in the positive spectrum, where $f(k)$ is negative. Hence for $A\in\mathbb{Z}$ and $\alpha<0$, $f(k)$ increases from $-\infty$ to zero in the  first gap and from zero to $\infty$ in every gap starting from the second one.

Denote by $-\Delta_{\balpha+\bgamma^{(1)}},\, \bgamma^{(1)}=\{\dots,0,\gamma_1,0,\dots\}$, the Hamiltonian corresponding to the system with one vertex perturbation of the strength $\gamma_1$; then the above considerations lead us to the following conclusion.

 % -------------- %
\begin{thm}[Magnetic case]
Assume that $A\notin\mathbb{Z}$. The essential spectrum of $-\Delta_{\balpha+\bgamma^{(1)}}$ coincides with the spectrum of $-\Delta_{\alpha}$.
For $\gamma_1<0$, the operator $-\Delta_{\alpha,\gamma_1}$ has precisely one simple impurity state in every odd gap of its essential spectrum (starting from the first one). On the other hand, for $\gamma_1>0$ it has precisely one simple impurity state in every even gap of its essential spectrum.
\end{thm}
 % -------------- %

\noindent As before, we single out the case of integer $A$.

\begin{thm}[Non-magnetic case]
Suppose that $A\in\mathbb{Z}$. The essential spectrum of $-\Delta_{\balpha+\bgamma^{(1)}}$ coincides with the unperturbed one. For $\alpha>0$ and $\gamma_1>0$, the operator $-\Delta_{\balpha+\bgamma^{(1)}}$ has no eigenvalues. On the other hand, for $\alpha>0$ and $\gamma_1<0$, it has precisely one simple impurity state in every gap of its essential spectrum. For $\alpha<0$ and $\gamma_1>0$, there is precisely one simple impurity state in every gap except the first one; for $\alpha<0$ and $\gamma_1<0$, there is precisely one simple eigenvalue below the first band and $-\Delta_{\balpha+\bgamma^{(1)}}$ has no other eigenvalues.
\end{thm}

%%%%%%%%%%%%%%%%%%%%%%%%%%%%%%%%%%%%%%%%%%%%%
\subsection{Example: two adjacent impurities}

Consider next the problem of eigenvalue existence in the case $m=2$ and denote the corresponding Hamiltonian by $-\Delta_{\balpha+\bgamma^{(2)}}$ where $\bgamma^{(2)}=\{\dots,0,\gamma_1,\gamma_2,0,\dots\}$. It is natural to assume that $\gamma_1\gamma_2\neq0$, since the opposite case has been already treated, and that $\gamma_1\neq\gamma_2$, since the case of equal impurities will be mentioned separately below, then Proposition~\ref{pro:EVCondition} leads to the a quadratic equation. One can easily solve it to deduce the relations
 % -------------- %
\begin{equation}\label{m=2}
\gamma_1+\gamma_2
=f(k)-
\frac{4k\cos A\pi}{\sin k\pi}
\Big(\xi(k)
\pm
\sqrt{
1+(\xi_1(k)-\xi_2(k))^2
}\Big)\,.
\end{equation}
 % -------------- %
Let $f_\pm$ stand for the right-hand side of relations~\eref{m=2} and suppose that $A\notin\mathbb{Z}$. Then, as $k^2$ varies from the lower end of a gap in $\sigma(-\Delta_\alpha)$ to the upper end, both $f_-(k)$ and $f_+(k)$ are continuous with respect to $k$, strictly increasing in with respect to $k^2$ and do not intersect. In particular, $f_-$ ($f_+$) increases from $-\infty$ to some positive (respectively, negative) number in the first spectral gap and then  from another positive (respectively, negative) number to $\infty$ in the second one (cf. Fig.~\ref{fig:m=2});
in the next two gaps the $f_\pm$ switch roles and so on. As in the case of a single impurity half of the spectral gaps is eliminated as an eigenvalue support when $A$ approaches an integer.

 % -------------- %
\begin{figure}[h!]
\centering
    \includegraphics[scale=0.33]{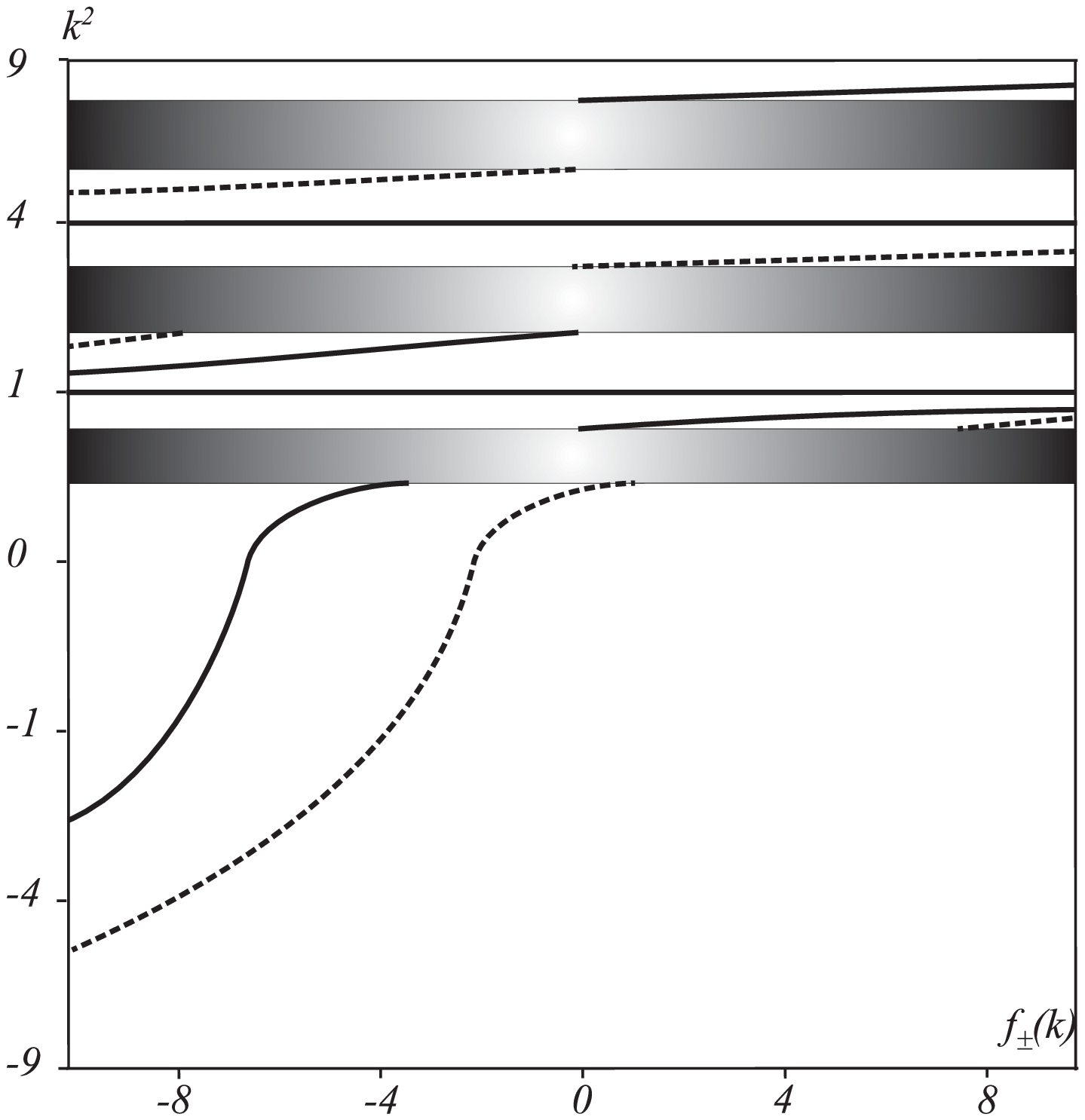}
    \includegraphics[scale=0.33]{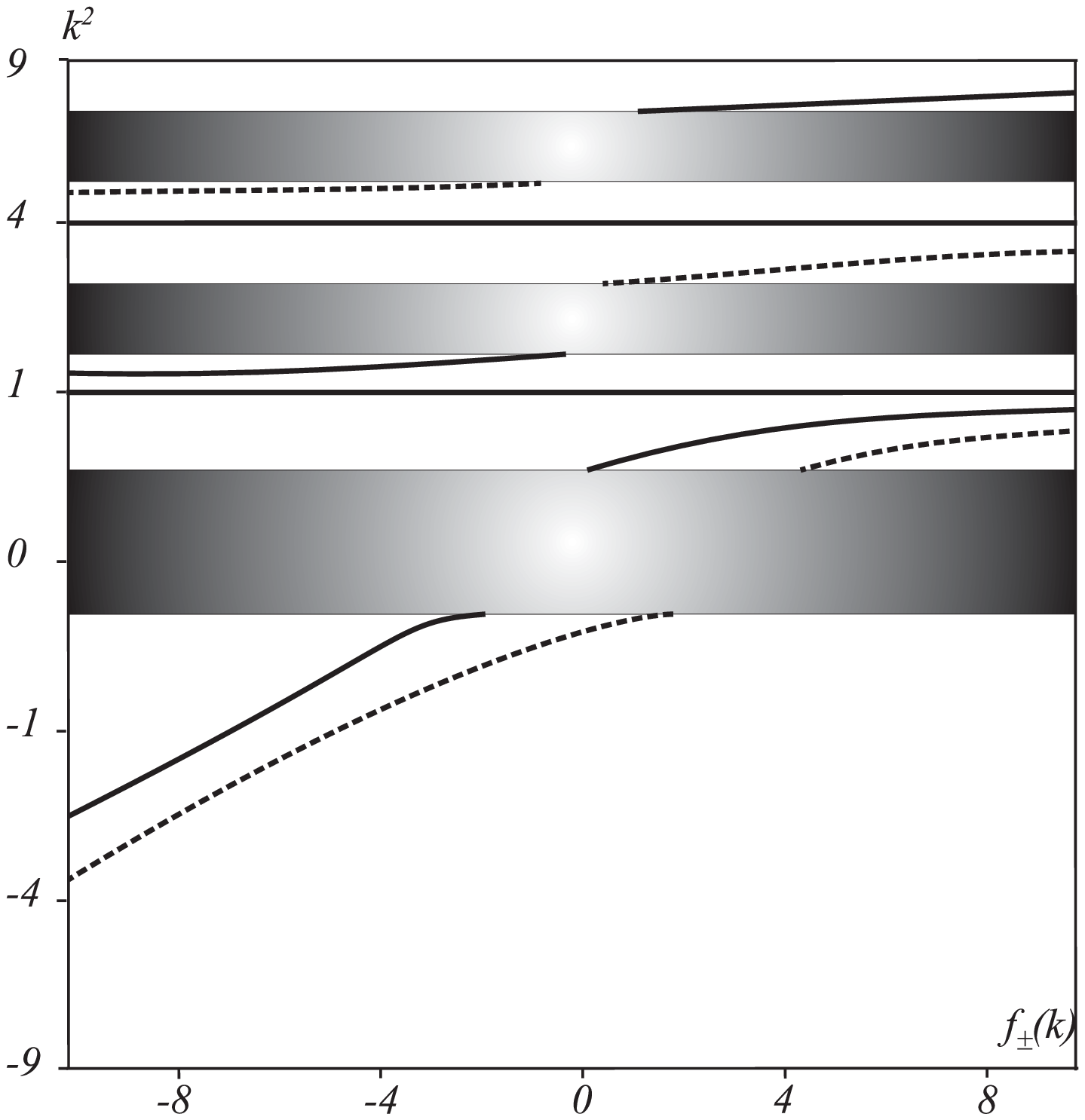}
    \includegraphics[scale=0.33]{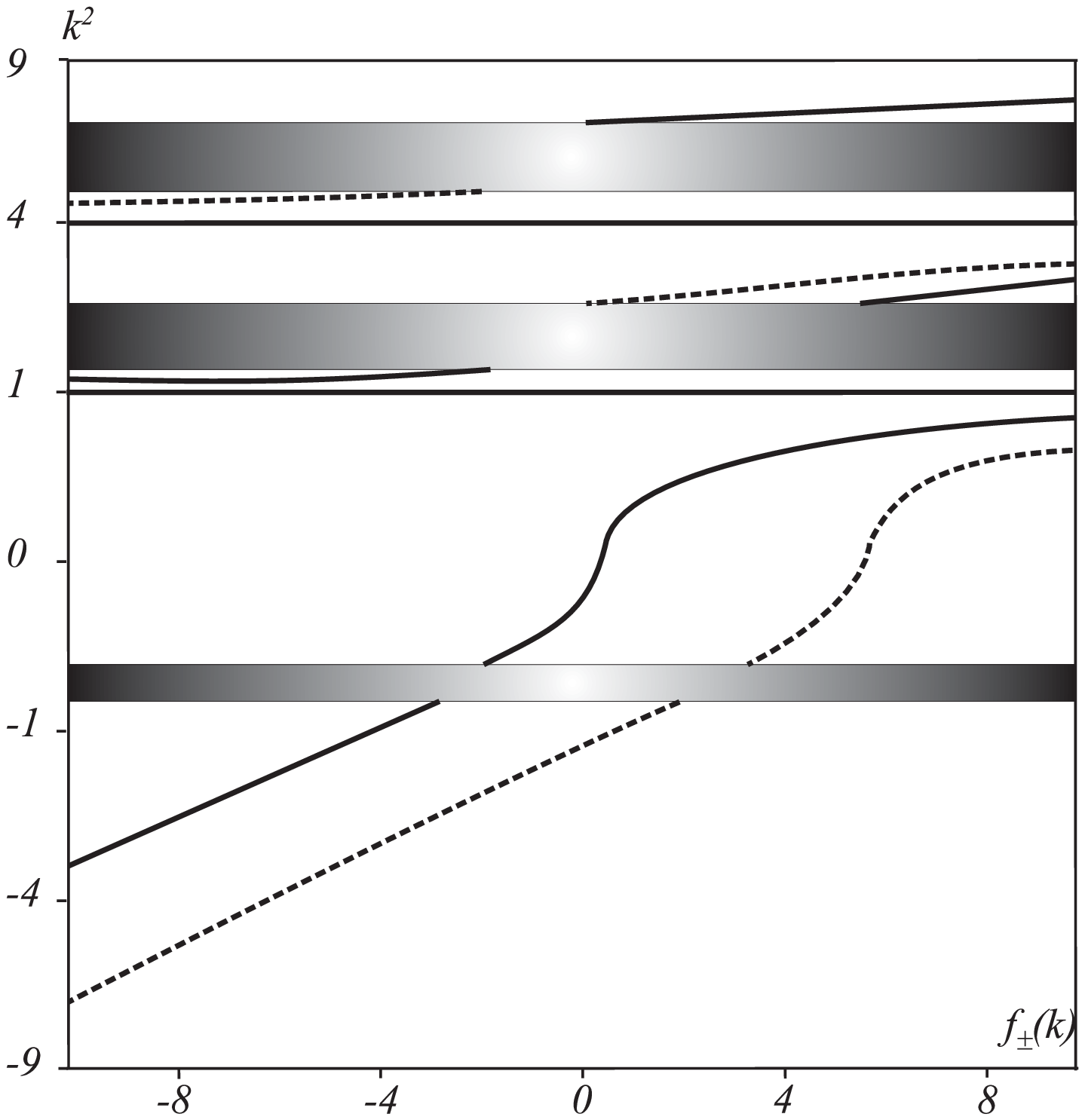}
  \caption{The dependence of $k^2$ on $f_\pm$ for $\cos A\pi=-0.6$, $\gamma_1=3$, $\gamma_2=1$, and the perturbation coupling strength (i) $\alpha=1$, (ii) $\alpha=-1$, (iii) $\alpha=-3$.}\label{fig:m=2}
\end{figure}
 % -------------- %

 % -------------- %
\begin{thm}[Magnetic case]
Suppose that $A\notin\mathbb{Z}$. The essential spectrum of $-\Delta_{\balpha+\bgamma^{(2)}}$ coincides with that of $-\Delta_{\alpha}$. If $\gamma_1+\gamma_2<0$, the operator $-\Delta_{\balpha+\bgamma^{(2)}}$ has at least one (and at most two) impurity state in every odd gap of its essential spectrum (starting from the first one). For $\gamma_1+\gamma_2>0$, it has at least one (and at most two) impurity state in every even gap of its essential spectrum.
\end{thm}
 % -------------- %

In the case of an integer $A$ we can make the following conclusion.

 % -------------- %
\begin{thm}[Non-magnetic case]
Suppose that $A\in\mathbb{Z}$. The essential spectrum of $-\Delta_{\balpha+\bgamma^{(2)}}$ coincides with that of $-\Delta_{\alpha}$. For $\alpha>0$ and a sufficiently large $\gamma_1+\gamma_2>0$, the operator $-\Delta_{\balpha+\bgamma^{(2)}}$ has no eigenvalues. For $\alpha>0$ and $\gamma_1+\gamma_2<0$, it has at least one impurity state in every gap of its essential spectrum. For $\alpha<0$ and a sufficiently large $\gamma_1+\gamma_2>0$ there is at least one impurity state in every gap except the first one. For $\alpha<0$ and $\gamma_1+\gamma_2<0$ there is at least one eigenvalue below the first band.
\end{thm}
 % -------------- %

%%%%%%%%%%%%%%%%%%%%%%%%%%%%%%%%%%%%%%%%%%
\subsection{Example: an array of identical impurities}

Let us consider again arbitrary $m\in\mathbb{N}$, however, restricting ourselves to the case when all the $\gamma_j$ are equal. We keep the notation $-\Delta_{\balpha+\bgamma}$ for the corresponding Hamiltonian. Abbreviating
 % -------------- %
\[
\cos\phi(k) = \xi_1(k)\,,
\]
 % -------------- %
we see that the matrix $\mathcal{N}_m$ of Proposition~\ref{pro:EVCondition} acquires the form
 % -------------- %
\[
\mathcal{N}_m(k)
=\frac{1}{\sin\phi(k)}
\left(\begin{array}{cc}
\sin[(m+1)\phi(k)] & -\sin [m\phi(k)]
\\[5pt]
\sin [m\phi(k)] & -\sin[(m-1)\phi(k)]
\end{array}\right)\,,
\]
 % -------------- %
and the characteristic equation \eref{cond:spectrumXi} has the following solutions
 % -------------- %
\begin{eqnarray*}
\lambda(k)
&=\cos\phi(k)
+\textstyle
\cot\Big[\frac{(m-1)\phi(k)}2\Big]\sin\phi(k)\,,
\\
\lambda(k)
&=\cos\phi(k)
-\textstyle
\tan\Big[\frac{(m-1)\phi(k)}2\Big]\sin\phi(k)\,.
\end{eqnarray*}
 % -------------- %
leading thus to the characteristic conditions
 % -------------- %
\begin{eqnarray}\nonumber
\gamma_1&=f(k)-\frac{4k\cos A\pi}{\sin k\pi}
\textstyle\cot\Big[\frac{(m-1)\phi(k)}2\Big]
\sin\phi(k)\,,
\\[-.5em]
\label{eq:gamma_1=...=gamma_m}
\\[-.5em]\nonumber
\gamma_1&=f(k)+\frac{4k\cos A\pi}{\sin k\pi}
\textstyle\tan\Big[\frac{(m-1)\phi(k)}2\Big]
\sin\phi(k)\,.
\end{eqnarray}
 % -------------- %
Note that if $m=2$, the above conditions lead to relations \eref{m=2} with $\gamma_2=\gamma_1$.

The explicit structure of the characteristic equations \eref{eq:gamma_1=...=gamma_m} shows that as $k^2$ varies from $-\infty$ to the lower end of $\sigma(-\Delta_\alpha)$, then both functions at the right-hand side of the above relations are continuous with respect to $k$, except possibly a finite number of points, strictly increasing in $k^2$ and their graphs do not intersect, in particular, at least one of the functions includes $(-\infty,0]$ in its range. Thus the following claim holds.

 % -------------- %
\begin{thm}
The essential spectrum of $-\Delta_{\balpha+\bgamma}$ coincides with that of $-\Delta_{\alpha}$. Suppose that $\gamma_1=\ldots=\gamma_m$, then $-\Delta_{\balpha+\bgamma}$ has at least one eigenvalue below the first spectral band.
\end{thm}
 % -------------- %

As we have mentioned both the functions at the right-hand side of \eref{eq:gamma_1=...=gamma_m} could have a finite number of jumps of the second order (approaching $\infty$ in the left vicinity of the discontinuity point and $-\infty$ in the right one). The total number of discontinuity points does not exceed $m-2$, hence the total number of eigenvalues in the first spectral gap does not exceed $m$. Of course, this fact follows also from general principles \cite{Wei80}, since the operators $-\Delta_{\alpha}$ and $-\Delta_{\balpha+\bgamma}$ have a common symmetric restriction with deficiency indices not exceeding $(m, m)$.

%%%%%%%%%%%%%%%%%%%%%%%%%%%%%%%%%%%%%%%%%%%%%%%%%%%%%%%%%%
%%  Weak coupling                                       %%
%%%%%%%%%%%%%%%%%%%%%%%%%%%%%%%%%%%%%%%%%%%%%%%%%%%%%%%%%%

\section{Weakly coupled systems}\label{sec:wcoupl}

Our next task is to compare the spectral properties of $-\Delta_\alpha$ with those produced by a weak finite-rank perturbation. Specifically, we suppose that the perturbation strength is $\eps\gamma_j,\, j=1,\dots,m$, at the vertices with the coordinates from $\pi\mathbb{M}$. Here $\eps$ is a small parameter and the perturbed Hamiltonian will be denoted by $-\Delta_{\balpha+\eps\bgamma}$. As one could expect such kind of perturbation preserves again the essential spectrum. In this section we are going to demonstrate that, as $\eps\to0$, the presence of the eigenvalue in the gap of the essential spectrum of $-\Delta_{\balpha+\eps\bgamma}$ is determined by the sign of $\sum\gamma_j$.

With this aim, we mimick the argument from the proof of Proposition~\ref{pro:EVCondition} to obtain a pair of characteristic equations,
 % -------------- %
\begin{eqnarray*}
\det\{\mathcal{N}_m(k) u_2(k),u_1(k)\}&=0\,,
\qquad\qquad\xi(k)<-1\,,
\\
\det\{\mathcal{N}_m(k) u_1(k),u_2(k)\}&=0\,,
\qquad\qquad\xi(k)>1\,,
\end{eqnarray*}
 % -------------- %
The product $\mathcal{N}_m=N_m\ldots N_1$ here can be written in terms of the matrices
 % -------------- %
\[
N_j(k)=N(k)+\frac{\eps\gamma_j\sin k\pi}{2k\cos A\pi}M
\]
 % -------------- %
with
 % -------------- %
\[
N(k)=
\left(\begin{array}{cc}
2\xi(k) & -1 \\ 1 & 0
\end{array}\right)\,,
\qquad
M=
\left(\begin{array}{cc}
1 & 0 \\ 0 & 0
\end{array}\right)\,.
\]
 % -------------- %
Using the explicit structure of the matrices $N_j$, one can easily see that, as $\eps\to0$, the product $\mathcal{N}_m$ admits the following expansion,
 % -------------- %
\[
\mathcal{N}_m(k)=(N(k))^m+\frac{\eps\sin k\pi}{2k\cos A\pi}\sum_{j\in\mathbb{M}}\gamma_j(N(k))^{m-j}M(N(k))^{j-1}+\OO(\eps^2)\,.
\]
 % -------------- %
Next, using the fact that $u_j$ is an eigenvector of the matrix $N$, we find the relations
 % -------------- %
\begin{eqnarray*}
\fl\det\{(N(k))^{m-j}M(N(k))^{j-1}u_2(k),u_1(k)\}
&=
(\lambda_2(k))^{m-1}\det\{Mu_2(k),u_1(k)\}\,,
\\\fl
\det\{(N(k))^{m-j}M(N(k))^{j-1}u_1(k),u_2(k)\}
&=
(\lambda_1(k))^{m-1}\det\{Mu_1(k),u_2(k)\}\,,
\end{eqnarray*}
 % -------------- %
which together with the above asymptotic formula for  $\mathcal{N}_m$ result in the relation
 % -------------- %
\begin{eqnarray*}
\det\{\mathcal{N}_m(k) u_2(k),u_1(k)\}=
(\lambda_2(k))^m\det\{u_2(k),u_1(k)\}
\\
\qquad
+
(\lambda_2(k))^{m-1}\det\{Mu_2(k),u_1(k)\}\frac{\eps\sin k\pi}{2k\cos A\pi}\sum_{j\in\mathbb{M}}\gamma_j+\OO(\eps^2)\,,
\end{eqnarray*}
 % -------------- %
or
 % -------------- %
\begin{eqnarray*}
\det\{\mathcal{N}_m(k) u_1(k),u_2(k)\}=
(\lambda_1(k))^m\det\{u_1(k),u_2(k)\}
\\
\qquad
+
(\lambda_1(k))^{m-1}\det\{Mu_1(k),u_2(k)\}\frac{\eps\sin k\pi}{2k\cos A\pi}\sum_{j\in\mathbb{M}}\gamma_j+\OO(\eps^2)
\end{eqnarray*}
 % -------------- %
as $\eps\to0$. Finally, from what has already been said, cf.~\eref{m=1}, it follows that the characteristic equation for $-\Delta_{\balpha+\eps\bgamma}$ reads as follows,
 % -------------- %
\begin{equation}\label{m>2}
\eps\sum_{j\in\mathbb{M}}\gamma_j+\OO(\eps^2)=f(k)\,,\qquad\eps\to0\,,
\end{equation}
 % -------------- %
where $f$ stands for the right-hand side of~\eref{m=1}. As a result we immediately obtain the following claim.

 % -------------- %
\begin{thm}[Magnetic case]
Suppose that $A\notin\mathbb{Z}$. For any $\eps\in(0,1)$ the essential spectrum of $-\Delta_{\balpha+\eps\bgamma}$ coincides with that of $-\Delta_\alpha$. Assume that $\sum_{j\in\mathbb{M}}\gamma_j<0$, then in the limit $\eps\to0$, the operator $-\Delta_{\balpha+\eps\bgamma}$ has exactly one simple impurity state in every odd gap of its essential spectrum. If the sum $\sum_{j\in\mathbb{M}}\gamma_j$ is positive, then in the limit $\eps\to0$  it has exactly one simple impurity state in every even gap of its essential spectrum. If $k_{n,\eps}^2$ is the corresponding eigenvalue, then $k_{n,\eps}$ admits the following asymptotic expansion
 % -------------- %
\[
k_{n,\eps}=k_n+(-1)^{2n+1}K_n\eps^2+\OO(\eps^2)\,,
\qquad\eps\to0\,,\qquad n\in\mathbb{N}\,.
\]
Here $k_1^2<k_2^2<\ldots$ are produced by the solutions to the equation $|\xi(k)|=1$, and
\[
	K_n
=
	\frac
{\sin k_n\pi\big(\sum_{j\in\mathbb{M}}\gamma_j\big)^2}
{\cos A\pi(32 k\pi_n^2-8\alpha k\pi_n\cot k_n\pi+8\alpha)}
\,,\qquad n\in\mathbb{N}\,.
\]
 % -------------- %
\end{thm}
 % -------------- %
\begin{proof}
The claims concerning existence of eigenvalues follow from the arguments above; the argument for the asymptotic expansion claim is rather straightforward, it suffices to use Taylor expansions at the right-hand side of~\eref{m>2}.
\end{proof}
 % -------------- %
\begin{thm}[Non-magnetic case]
Suppose that $A\in\mathbb{Z}$. The essential spectrum of $-\Delta_{\balpha+\eps\bgamma}$ coincides with that of $-\Delta_{\alpha}$. To describe the discrete spectrum $-\Delta_{\alpha+\eps\bgamma}$ in the limit $\eps\to0$ we distinguish four cases depending on the sign of coupling constant $\alpha$ and the sum $\sum_{j\in\mathbb{M}}\gamma_j$. First we assume that both are positive, then $-\Delta_{\balpha+\eps\bgamma}$ has no eigenvalues as $\eps\to0$.
For $\alpha>0$ and $\sum_{j\in\mathbb{M}}\gamma_j<0$, the operator $-\Delta_{\balpha+\eps\bgamma}$ has precisely one simple impurity state in every gap of its essential spectrum. For $\alpha<0$ and $\sum_{j\in\mathbb{M}}\gamma_j>0$ there is precisely one simple impurity state in every gap except the first one. Finally, if we assume that both $\alpha$ and $\sum_{j\in\mathbb{M}}\gamma_j$ are negative, then there is precisely one simple eigenvalue below the first band and $-\Delta_{\balpha+\eps\bgamma}$ has no other eigenvalues.
\end{thm}

%%%%%%%%%%%%%%%%%%%%%%%%%%%%%%%%%%%%%%%%%%%%%%%%%%%%%%%%%%
%%  Distant perturbations                               %%
%%%%%%%%%%%%%%%%%%%%%%%%%%%%%%%%%%%%%%%%%%%%%%%%%%%%%%%%%%

\section{Systems with distant impurities}
\label{sec:dpert}

Finally, let us consider the system with two impurities at large distances from each other. To be more specific, we change the coupling constants at two arbitrary but fixed points into $\alpha+\gamma_1$ and $\alpha+\gamma_2$; we suppose that there are exactly $n$ graph vertices between the chosen two. For the sake of brevity, let $-\Delta_{\alpha,n}$ denote the Hamiltonian of the perturbed system; we are interested in spectral properties of the operator $-\Delta_{\alpha,n}$ for large $n$. To this aim, we repeat the proof of Proposition~\ref{pro:EVCondition} to obtain one of the characteristic equations
 % -------------- %
\begin{eqnarray*}
\det\{\mathcal{N}_n(k) u_2(k),u_1(k)\}=0
\end{eqnarray*}
 % -------------- %
or
 % -------------- %
\begin{eqnarray*}
\det\{\mathcal{N}_n(k) u_1(k),u_2(k)\}=0\,,
\end{eqnarray*}
 % -------------- %
depending on $\sgn(\xi(k))$, where the product $\mathcal{N}_n$ is defined as follows
 % -------------- %
\[
\mathcal{N}_n(k)
=
\underbrace{\Big(
	N(k)+\frac{\gamma_2\sin k\pi}{2k\cos A\pi}M
\Big)}_{N_2(k)}
(N(k))^n
\underbrace{\Big(
	N(k)+\frac{\gamma_1\sin k\pi}{2k\cos A\pi}M
\Big)}_{N_{1}(k)}
\]
 % -------------- %
with the definitions of the matrices $N$ and $M$ given above. Noting that
 % -------------- %
\begin{eqnarray*}
N_i(k)u_j(k)
	&=
c_{1j}(k,\gamma_i)\lambda_1(k)u_1(k)
	+
c_{2j}(k,\gamma_i)\lambda_2(k)u_2(k)\,,
\end{eqnarray*}
 % -------------- %
with
 % -------------- %
\begin{equation}\label{matr:C}
c_{ij}(k,\gamma):=\delta_{ij}+
{(-1)^i\gamma\,\sgn(\xi(k))}/{f(k)}\,,
\end{equation}
 % -------------- %
where $\delta_{ij}$ is the Kronecker delta and $f$ stands for the right-hand of \eref{m=1}. We recall that $u_j$ is an eigenvector of $N$ corresponding to the eigenvalue $\lambda_j$ to conclude that
 % -------------- %
\begin{eqnarray*}
\fl
\det\{\mathcal{N}_n(k) u_2(k),u_1(k)\}&=
\big[(\lambda_1(k))^nc_{12}(k,\gamma_1)c_{21}(k,\gamma_2)
\\\fl
&\quad
+
(\lambda_2(k))^nc_{22}(k,\gamma_1)c_{22}(k,\gamma_2)\big]
\det\{u_2(k),u_1(k)\}\,,
&
\\\fl
\det\{\mathcal{N}_n(k) u_1(k),u_2(k)\}
&=
\big[(\lambda_1(k))^nc_{11}(k,\gamma_1)c_{11}(k,\gamma_2)
\\\fl
&\quad+
(\lambda_2(k))^nc_{12}(k,\gamma_1)c_{21}(k,\gamma_2)\big]
\det\{u_1(k),u_2(k)\}\,.
\end{eqnarray*}
 % -------------- %
Using the expression for the entries $c_{ij}$, the characteristic equation for our system reads
 % -------------- %
\begin{equation}\label{2farImp}
(f(k)/\gamma_1-1)(f(k)/\gamma_2-1)=
(\lambda(k))^{2n+2}\,.
\end{equation}
 % -------------- %
Consider an arbitrary but fixed  spectral gap of $-\Delta_\alpha$, denoted $\mathcal{I}$, which is an open subset of the real line.
We have to analyze equation~\eref{2farImp} as $k^2$ varies from the lower end of $\mathcal{I}$ to its upper end. It is worth noting that the right-hand side of the above relation approaches one, as $k^2$ approximates any of the endpoints of the spectral gap, and tends to zero as $n\to\infty$  within the interval $\mathcal{I}$.
At the same time, the left-hand side also approaches one, as $k^2$ approximates the endpoints of $\mathcal{I}$.

For starters, we assume that $\gamma_1$ and $\gamma_2$ have different signs and recall that $f$ preserves its sign on $\mathcal{I}$. Hence the range of the left-hand side on $\mathcal{I}$ includes all values from $-\infty$ to one. Moreover, having the explicit expression for $f$ one can show
that if $f$ is negative on $\mathcal{I}$, then the left-hand side of \eref{2farImp} is monotonously increasing function on $\mathcal{I}$, possibly except a small left neighborhood of its right endpoint. Similarly, for positive $f$ the left-hand side monotonously decreases on $\mathcal{I}$, possibly except a small right neighborhood of its left endpoint. In both cases the equation $(f(k)/\gamma_1-1)(f(k)/\gamma_2-1)=0$ admits a unique solution on $\mathcal{I}$. As a result, equation~\eref{2farImp} gives a unique solution in every such an interval, thus producing an eigenvalue in every spectral gap of the operator $-\Delta_\alpha$.

Next we address the situation when $\gamma_1$ and $\gamma_2$ are of the same sign and distinguish two cases depending on the sign of $f$. Suppose first that the sign of the function $f$ on $\mathcal{I}$ is opposite to those of $\gamma_j$. Then the left-hand side of \eref{2farImp} is a monotonous function on $\mathcal{I}$ and its range there is $(1,\infty)$, hence \eref{2farImp} admits no solution on $\mathcal{I}$. Let us turn to the second case, when $f$ and $\gamma_j$ are of the same sign. Then, obviously, the equation $(f(k)/\gamma_1-1)(f(k)/\gamma_2-1)=0$ has two solutions on $\mathcal{I}$, say, $x_1$ and $x_2$ (in the case $\gamma_1=\gamma_2$ we have $x_1=x_2$). In addition, depending on the sign of $f$ the left-hand side of \eref{2farImp} monotonously decreases from one (from $\infty$) to its local minimum, which is negative (or to zero if $\gamma_1=\gamma_2$) and attained at some point between $x_1$ and $x_2$, and then monotonously decreases from its local minimum to $\infty$ (respectively, to one). In this case equation~\eref{2farImp} admits two solutions on $\mathcal{I}$. If $\gamma_1=\gamma_2$, then the mentioned solutions are close to each other. Indeed, equation \eref{2farImp} leads to the relations
 % -------------- %
\[
\gamma_1
	=
	\frac{f(k)}{1\pm|\lambda(k)|^{n+1}}
\]
 % -------------- %
and from the explicit formulas for $f$ and $\lambda$ it follows that, as $n\to\infty$, both functions at the right-hand side are $\OO(\mathrm{exp}\{-C_1(n+1)\})$-close to each other for some positive $C_1$. At the same time, one can see that the functions are not too sloping, hence that both  solutions are $\OO(\mathrm{exp}\{-C_2(n+1)\})$-close with possibly different constant $C_2>0$. Summarizing the discussion, we have arrived at the following conclusions.

 % -------------- %
\begin{thm}[Magnetic case]
Suppose that $A\notin\mathbb{Z}$. For any $n\in\mathbb{N}$ the essential spectrum of $-\Delta_{\alpha,n}$ coincides with that of $-\Delta_\alpha$.
Assume that $\gamma_1\gamma_2<0$, then any for sufficiently large $n$, the operator $-\Delta_{\alpha,n}$ has precisely one simple impurity state in every gap of its essential spectrum. If $\gamma_1$ and $\gamma_2$ are positive (negative), then for sufficiently large $n$, $\:-\Delta_{\alpha,n}$ has two simple impurity states in every even (respectively, odd) gap of its essential spectrum and no impurity state in every odd (respectively, even) one (provided we start counting from the first gap). If $\gamma_1=\gamma_2$, then impurity states in every even or odd gap are exponentially close to each other with respect to $n$.
\end{thm}
 % -------------- %
\begin{thm}[Non-magnetic case]
Suppose that $A\in\mathbb{Z}$. For any $n\in\mathbb{N}$ the essential spectrum of $-\Delta_{\alpha,n}$ coincides with that of $-\Delta_\alpha$.
Assume that $\gamma_1\gamma_2<0$, then for any sufficiently large $n$, the operator $-\Delta_{\alpha,n}$ has precisely one simple impurity state in every gap of its essential spectrum. If $\gamma_1$ and $\gamma_2$ are positive (negative), then for sufficiently large $n$, $\:-\Delta_{\alpha,n}$ has two simple impurity states in every odd (respectively, even) gap of its essential spectrum and no impurity state in every even (respectively, odd) one (provided we start counting from the zeroth gap). If $\gamma_1=\gamma_2$, then mentioned two impurity states are exponentially close to each other with respect to $n$.
\end{thm}

%%%%%%%%%%%%%%%%%%%%%%%%%%%%%
\subsection*{Acknowledgments}

The research was supported by the Czech Science Foundation (GA\v{C}R) within the project 14-06818S and by the European Union with the project ``Support for research teams on CTU'' CZ.1.07/2.3.00/30.0034.

%%%%%%%%%%%%%%%%%%%%%%%%


\subsection*{References}

\begin{thebibliography}{11}

 % -------------- %
\bibitem[Ag82]{Ag82}
S.~Agmon: \emph{Lectures on Exponential Decay of Solutions of Second-Order Elliptic Equations}, Princeton University Press, Princeton 1982.
 %----------------%
\bibitem[BK13]{BK13}
G. Berkolaiko, P. Kuchment: \emph{Introduction to Quantum Graphs}, Amer. Math. Soc., Providence, R.I., 2013.
 % -------------- %
\bibitem[Ca97]{Ca97}
C.~Cattaneo: The spectrum of the continuous Laplacian on a graph, {\it Monatsh.Math.} {\bf 124} (1997), 215--235.
 % -------------- %
\bibitem[CP14]{CP14}
T.~Cheon, S.~S.~Poghosyan: Exotic quantum transport in double-stranded Kronig-Penney model, \texttt{arXiv:1410.8647 [quant-ph]}
 % -------------- %
\bibitem[DET08]{DET08}
P.~Duclos, P.~Exner, O.~Turek: On the spectrum of a bent chain graph, {\it J.~Phys.~A: Math.~Theor.} {\bf 41} (2008), 415206 (18pp).
 % -------------- %
\bibitem[Ex97]{Ex97}
P.~Exner: A duality between Schr\"odinger operators on graphs and certain Jacobi matrices, {\it Ann. Inst. H.~Poincar\'{e} A:  Phys. Th\'eor.} {\bf 66} (1997), 359--371.
 % -------------- %
\bibitem[EKW10]{EKW10}
P.~Exner, P.~Kuchment, B.~Winn: On the location of spectral edges in $\mathbb{Z}$-periodic media, {\it J. Phys. A: Math. Theor.} {\bf 43} (2010), 474022 (8pp).
 % -------------- %
\bibitem[KS03]{KS03}
V. Kostrykin, R. Schrader: Quantum wires with magnetic fluxes, {\it Comm. Math. Phys.} {\bf 237} (2003), 161-–179.
 % -------------- %
\bibitem[Pa13]{Pa13}
K.~Pankrashkin: An example of unitary equivalence between self-adjoint extensions and their parameters, {\it J. Funct. Anal.}  {\bf 265} (2013), 2910--2936.
 % -------------- %
\bibitem[Si76]{Si76}
B.~Simon: The bound state of weakly coupled Schr\"{o}dinger operators in one and two dimensions, {\it Ann. Phys.} {\bf 97} (1976), 279--288.
 % -------------- %
\bibitem[We80]{Wei80}
J.~Weidmann: \emph{Linear Operators in Hilbert Space}, Springer, New York 1980.

\end{thebibliography}
\end{document}